\documentclass[manuscript,screen]{acmart}

\usepackage{amsfonts}   
\usepackage{amsmath}    
\usepackage{amsthm}         
\usepackage{graphicx}       
\usepackage{subfigure}
\usepackage{hyperref}       
\usepackage{algorithm}
\usepackage{algorithmicx}
\usepackage{algpseudocode}  
\usepackage{courier}        

\AtBeginDocument{%
  \providecommand\BibTeX{{%
    \normalfont B\kern-0.5em{\scshape i\kern-0.25em b}\kern-0.8em\TeX}}}

\setcopyright{acmcopyright}
\copyrightyear{2022}
\acmYear{2022}
\acmDOI{10.1145/3546000.3546027}

\acmConference[HP3C '22]{HP3C '22: International Conference on High Performance Compilation, Computing and Communication}{June 23--25, 2022}{Jilin, China}
\acmBooktitle{HP3C '22: International Conference on High Performance Compilation, Computing and Communication,
  June 23--25, 2022, Jilin, China}
\acmPrice{15.00}
\acmISBN{978-1-4503-9629-5}

\theoremstyle{thmstyleone}%
\newtheorem{theorem}{Theorem}
\theoremstyle{thmstyletwo}%

\theoremstyle{thmstylethree}%
\newtheorem{definition}{Definition}%

\begin{document}

\title{A Local-Ratio-Based Power Control Approach for Capacitated Access Points in Mobile Edge Computing}

\author{Qinghui Zhang}
\affiliation{%
  \institution{Yunnan University}
  \city{Kunming}
  \country{China}}
\email{zhangqinghui@mail.ynu.edu.cn}

\author{Weidong Li}
\affiliation{%
  \institution{Yunnan University}
  \city{Kunming}
  \country{China}}
\email{weidong@ynu.edu.cn}

\author{Qian Su}
\affiliation{%
  \institution{Yunnan University}
  \city{Kunming}
  \country{China}}
\email{suqian@ynu.edu.cn}

\author{Xuejie Zhang}
\authornotemark[1]
\affiliation{%
  \institution{Yunnan University}
  \city{Kunming}
  \country{China}}
\email{xjzhang@ynu.edu.cn}

\begin{abstract}
  Terminal devices (TDs) connect to networks through access points (APs) integrated into the edge server. This provides a prerequisite for TDs to upload tasks to cloud data centers or offload them to edge servers for execution. In this process, signal coverage, data transmission, and task execution consume energy, and the energy consumption of signal coverage increases sharply as the radius increases. Lower power leads to less energy consumption in a given time segment. Thus, power control for APs is essential for reducing energy consumption. Our objective is to determine the power assignment for each AP with same capacity constraints such that all TDs are covered, and the total power is minimized. We define this problem as a \emph{minimum power capacitated cover } (MPCC) problem and present a \emph{minimum local ratio} (MLR) power control approach for this problem to obtain accurate results in polynomial time. Power assignments are chosen in a sequence of rounds. In each round, we choose the power assignment that minimizes the ratio of its power to the number of currently uncovered TDs it contains. In the event of a tie, we pick an arbitrary power assignment that achieves the minimum ratio. We continue choosing power assignments until all TDs are covered. Finally, various experiments verify that this method can outperform another greedy-based way.
\end{abstract}




\maketitle

\section{Introduction}\label{sec1}
With the rapid development of beyond 5G/6G and the Internet of Things, increasing number of terminal devices (TD) are being deployed at the edge of networks \cite{Zhang2022}. Even with advances in network technology, data centers cannot guarantee acceptable latency and network unobstructed. By placing the computing and storage resources closer to the TDs, mobile edge computing (MEC) can significantly increase performance in terms of low latency, reduced communications overhead, and high-quality user experience\cite{Abbas2018}.

Nowadays, many tasks of TDs will be prioritized to migrate to servers close to TDs for computing or storage, unless these tasks can only be handled in cloud data centers. By the way, a TD usually establishes network connection with edge server through its access point (AP, it can be an integral part of the edge server itself). This provides a prerequisite for TDs to upload tasks to cloud data centers or offload them to edge servers for execution. During these processes, signal coverage, data transmission, task execution, etc., cost energy. Therefore, in recent years, considerable research on reducing energy consumption has been conducted, and power control is an important topic.

In typical wireless networks, the signal coverage of an AP, which is determined by its power, is a disk area centered on it. Because the signal intensity will attenuate with distance, a larger signal coverage disk needs more power\cite{Ronnow2019}. Many power control methods have been proposed for different wireless networks. About power control of cellular network, \cite{Trinh2020} considered the sum spectral efficiency (SE) optimization problem in multi-cell Massive MIMO systems with a varying number of active users. The neural network they proposed, PowerNet, only uses the large-scale fading information to predict both the pilot and data powers. In \cite{Dai2021}, Dai \emph{et al.} investigated the joint optimization of BS clustering and power control for NOMA-enabled CoMP transmission in dense cellular networks to maximize system sum-rate. In addition, the scope of \cite{Chincoli2018} was to investigate how machine learning may be used to bring wireless nodes to the lowest possible transmission power level and, in turn, to respect the quality requirements of the overall network. Lowering transmission power has benefits in terms of both energy consumption and interference.

However, both upload bandwidth and server's computing resources are limited. This provides us a new perspective to power control for edge wireless network. In \cite{Wu2021}, Wu \emph{et al.} first studied the joint multiuser offloading and transmission power control optimization problem in a multi-channel wireless interference scenario. They then proposed an efficient semi-distributed algorithm consisting of two subalgorithms of offloading scheme and transmission power control.

All the methods stated above guarantee QoS will not deteriorate due to interference. The same is true in this paper. We build a signal coverage model simplified communication process and paid attention to the location of facilities and the capacity of AP. Our objective is to determine the power assignment in MEC such that APs with limited capacity can cover all TDs and minimize the total power. We define this problem as the \emph{minimum power capacitated cover} (MPCC) problem, which is an NP-hard problem. A local ratio approach is proposed to solve it. Finally, various experiments verify that this method can outperform another greedy-based way.

\subsection{Related Works}
The problem of resource allocation has been paid much attention in various fields, and it has many different optimization objectives. Furthermore, resource allocation problems are usually NP-hard problems, and how to get satisfactory results in a short time is a considerable challenge. Aiming at social welfare maximization, reference \cite{10.1007/s00607-021-00962-6} proposed an auction mechanism for resource allocation in the Cloud Edge collaboration framework based on live video webcast services. Reference \cite{8241849} addressed the problem of heterogeneous physical machines resource management (HPMRM) in heterogeneous clouds. Trade-offs between energy and performance are important for energy-aware scheduling. In reference \cite{LI201970}, a polynomial-time algorithm was designed for energy-aware profit maximizing scheduling problem.

MPCC is a fundamental \emph{minimum power coverage} (MPC) problem that has been extensively studied in the past two decades. Originally, the problem was introduced for modeling the power control problem in the communication of edge computing. Broadly, MPCC belongs to the family of \emph{minimum weight set cover} (MWSC) problems. MPCC can be viewed as a special variant of the MWSC that has the following general formulation: Given a universe $E$ of elements (points or terminal devices) and a family ${\mathcal G}$ of ranges (or geometric objects like disks and polygons) where each range has a weight assignment $w:\mathcal{G} \mapsto {\mathbb{R}^ + }$, determine a set $\mathcal{F} \subseteq \mathcal{G}$ of ranges such that every element in $E$ is covered by at least one range in $\mathcal{F}$ and that the total weight $\sum\nolimits_{G \in \mathcal{F}} {w(G)} $ is minimized. Clearly, MPCC can be formulated as a special MWSC, where $\mathcal{G}$ is the set of all possible disks centered at some APs  having at least one TD on their boundaries. The difference is that the capacity of each $a \in A$ is limited to $k(a)$; therefore, the capacity of the disk centered at $a$ in $\mathcal{G}$ is also constrained.

The MWSC problem is, in general,  challenging to  solve optimally, even for some simple versions. For example, Alt  \emph{et al.} and Bil\`{o} \emph{et al.} present a minimum cost covering problem without capacity constraint in \cite{Alt2006, Bilo2005}, which is still NP-hard for any $\alpha  > 1$. Thus, finding polynomial time approximation algorithms is the main objective for MPCC problems.

For the \emph{geometric minimum weight set cover} problem, it was studied from different aspects. For the minimum disk cover problem, in which disks may have different sizes, Mustafa and Ray designed a PTAS using a local search method \cite{Mustafa2010}.  Considering weight, Varadarajan \cite{Varadarajan2010} presented a clever quasi-uniform sampling technique that was improved by Chan \emph{et al.} \cite{Chan2012}, yielding a constant approximation for the minimum weight disk cover problem. About partial cover problem, Liu \emph{et al.} \cite{Liu2021} introduced the $k$-prize-collecting minimum power cover problem ($k$-PCPC), and present a novel two-phase primal-dual algorithm and got an approximation ratio of at most $3^{\alpha}$. The author then introduced the minimum power cover problem with submodular and linear penalties in \cite{Liu2022} and presented a combinatorial primal-dual $(3^{\alpha} + 1)$-approximation algorithm.

In previous studies, the minimum power cover problem has rarely been considered for capacity constraints. However, for the vertex cover problem, another research topic of the set cover problem, many studies consider the minimum capacitated vertex cover problem. Consider a graph $G = (V,E)$ with weights and capacities on the vertices. All edges must be covered by vertices under the capacity constraint, and the objective is to minimize the weight. \cite{Guha2003} give a primal-dual-based approximation algorithm with an approximation guarantee of 2 (every edge has two vertices). The study of the vertex cover problem with hard capacity constraints (VC-HC) was initiated in \cite{Chuzhoy2006}. They established a surprising result that, while this setting admits constant factor approximations, it becomes set-cover hard when $\left\{ {0,1} \right\}$-weighted vertices are considered. \cite{Kao2021} considered VC-HC on hypergraphs. This literature, which considers a hypergraph $G$ with a maximum edge size of $f$, provides an iterative partial rounding technique and obtains an $f$-approximation, improving upon the previous results of \cite{Cheung2014}.
\section{System model}
In this section, we present the system model and a corresponding integer programming model. We focus on power control with capacitated APs and covering all TDs to minimize total power.

\subsection{Model Description}
Consider $m$ APs with same capacity and $n$ TDs randomly distributed in 2-dimensional space. Each TD accesses the edge network through its covering AP. Let $(U,A,k)$ be an instance of the MPCC problem with $U = \left\{ {1,2,...,n} \right\}$ and $A = \left\{ {1,2,...,m} \right\}$, and let $k \in {\mathbb{Z}^ + }$ denote the capacity of IP to each AP $a \in A$. Each AP $a$ can adjust its own power, and the relationship between the power $p(a)$ and the radius $r(a)$ of its service area is
\begin{equation}
    p(a) = c \cdot r{(a)^\alpha },
  \label{eq:1}
\end{equation}
where $c$ and $\alpha$ are constants ($\alpha$ is usually called the \emph{attenuation factor}). The center and radius can be used to define a disk, so $a$ and $u \in U$ can determine a unique disk ${D_{au}}$ whose radius is precisely equal to the distance between $a$ and $u$ ($u$ is on the boundary of ${D_{au}}$ in the optimal case). Therefore, at most $mn$ disks must be considered. We denote the set of such disks by $\mathcal{D}$. Each instance corresponds to a $\mathcal{D}$, and we use $D$ to represent any disk in $\mathcal{D}$. The set of disks with $a$ as the center is a subset of $\mathcal{D}$, denoted by $\mathcal{D}(a)$. To simplify the notation, we use $D$ to represent both a disk in $\mathcal{D}$ and the set of TDs contained in D and use $r(D)$ and $p(D)$ to denote the radius and power of disk $D$, where $p(D) = c \cdot r{(D)^\alpha }$. The mapping $k( \cdot ):\mathcal{D} \mapsto k$ can determine the capacity of $D \in \mathcal{D}$. Notably, $u$ \emph{contained} in $D$ means that $u$ is within the range of $D$ ($u \in D$ ) and $u$ is \emph{covered} by $a$ means that $u$ gains an IP of $a$ to access the edge network (the former does not occupy the resources of the corresponding AP). Clearly, the capacity of $D \in \mathcal{D}(a)$ formed by different power strategies of one AP is equal and changes together.  Table \ref{table:notation} summarizes the notation for reference.

Fig. \ref{fig1} shows an instance with three APs and six TDs. The capacities of AP1, AP2 and AP3 are all 2, respectively (the first number in brackets indicates the AP index, the following represents its capacity). The disk formed by AP1 contains TDs 1, 2 and 4, TDs 3,5 is contained by the disk formed by AP2, and TDs 4, 5 and 6 are contained by AP3's disk. Notably, TD4 and TD5 are contained by two disks simultaneously. Since the capacity of all APs are only 2, TD4 is eventually covered by AP3 and TD5 is covered by AP2.
\begin{figure}[htbp]
  \centering
  \includegraphics[width=6cm]{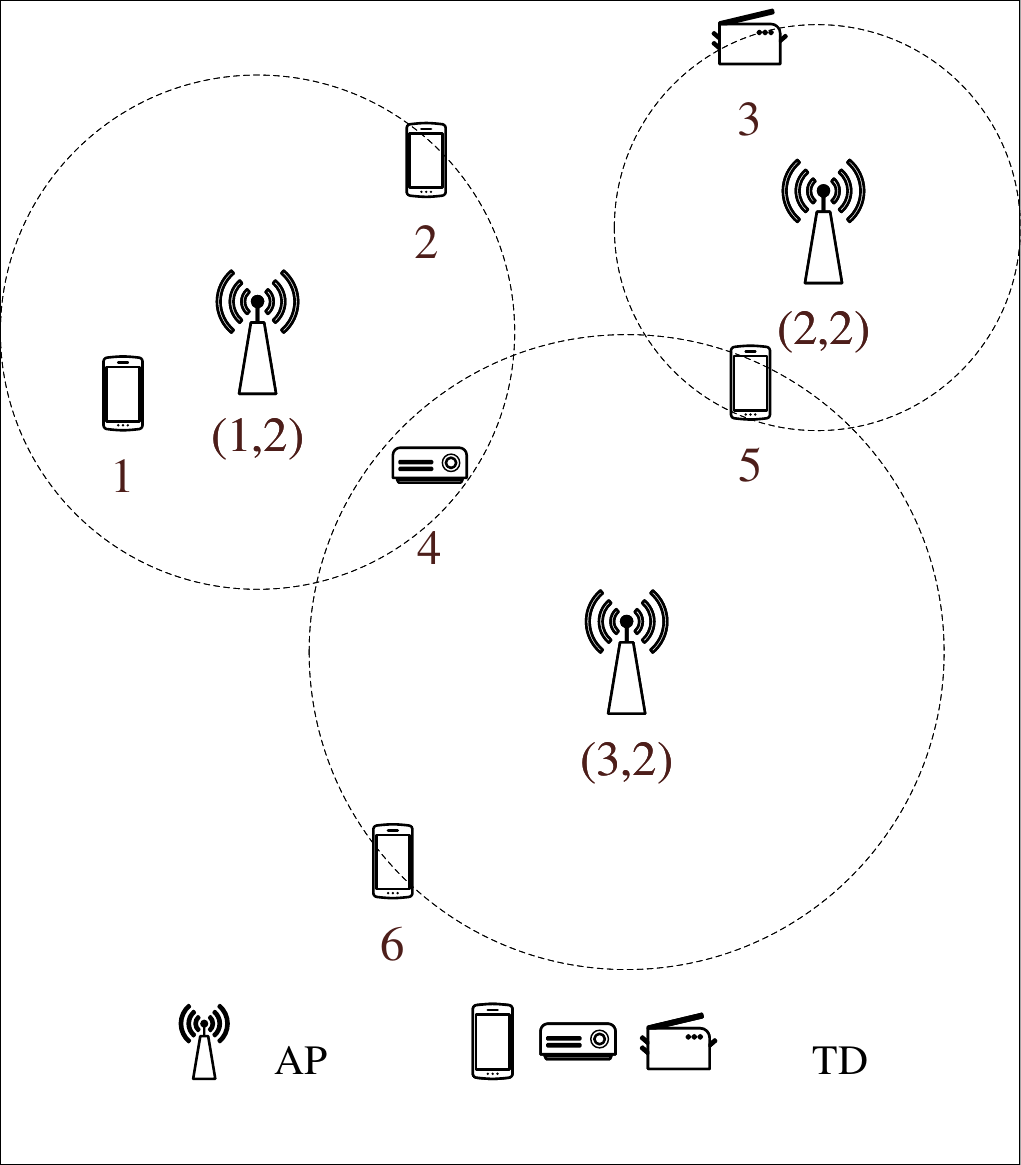}
  \caption{An instance of the MPCC problem.}
  \label{fig1}
\end{figure}

In practical situations, a TD may locate outside the maximum coverage that all APs can form, or the number of TDs exceeds the total capacity of all APs. We assume that an AP has no upper power limit and that the entire system has enough capacity to cover all the TDs. In this way, we can obtain a scheme covering all TDs through the technique in this paper, which can provide a reference for the later decision on whether to cover some distant TDs. We can also solve the instance where the demand exceeds the capacity by using the method of this paper many times. Therefore, the above two assumptions are reasonable.
\subsection{Integer Program for the MPCC Problem}
For a subset $\mathcal{\bar D}$ of $\mathcal{D}$, $\mathcal{\bar D}$ is a feasible solution to the MPCC problem and must meet the following conditions:
\begin{enumerate}
  \item \textbf{Full coverage condition.}
      $\mathcal{C}( \mathcal{\bar D}) = \bigcup\nolimits_{D \in \mathcal{\bar D}} D$ denotes the TDs covered by $\mathcal{\bar D}$; then, $\mathcal{C}( \mathcal{\bar D}) = U$ must be established.
  \item \textbf{Limited capacity condition.} Since the IP capacity of each AP is limited, the number of TDs covered by an AP must be less than its capacity.
  \item \textbf{Power uniqueness.} For any $a \in A$, there are $n$ power strategies, but at most one strategy is selected; that is, there is at most one disk with $a$ as the center in $\mathcal{\bar D}$.
\end{enumerate}
According to the above three conditions, we can formulate an integer programming model for the MPCC problem. Variable ${z_D} \in \left\{ {0,1} \right\}$ indicates whether $D \in \mathcal{D}$ is selected; that is, ${z_D} = 1$ if and only if $D$ is selected. Variable ${y_{uD}} \in \left\{ {0,1} \right\}$ indicates whether $u \in U$ is covered by $D$, where ${y_{uD}} = 1$ if and only if $u$ is covered by $D$. The following is the integer program:
\begin{align}\label{ip:2}
  \min  & \sum\limits_{D \in \mathcal{D}} {p\left( D \right)}  \cdot {z_D} \\
  \label{ip:2.1}s.t. & \sum\limits_{D:u \in D} {{y_{uD}}}  \ge 1,u \in U \\
  \label{ip:2.2}& {k(D)}{z_D} - \sum\limits_{u:u \in D} {{y_{uD}}}  \ge 0,D \in \mathcal{D}\\
  \label{ip:2.3}& {z_D} \ge {y_{uD}},u \in D \in \mathcal{D}\\
  \label{ip:2.4}& \sum\limits_{D \in \mathcal{D}(a)} {{z_D}}  \le 1,a \in A\\
  \label{ip:2.5}& {z_D} \in \left\{ {0,1} \right\},{y_{uD}} \in \left\{ {0,1} \right\},\forall u \in U,\forall D \in \mathcal{D}
\end{align}

Constraint (\ref{ip:2.1}) corresponds to the full coverage condition, (\ref{ip:2.2}) corresponds to the limited capacity condition, and (\ref{ip:2.4}) corresponds to the power uniqueness. Constraint (\ref{ip:2.3}) is an implicit condition in this integer program. If $u$ is covered by $D$ in the final result set, then $D$ must be selected as the power strategy of its center AP.

\begin{table}[t!]
    \caption{Notation}
    \centering
    \label{table:notation}
    \newcommand{\tabincell}[2]{\begin{tabular}{@{}#1@{}}#2\end{tabular}}
    \begin{tabular}{lll}
        \hline
        \textbf{Notation} & \textbf{Description}\\ \hline
        $z_D$ & \tabincell{l}{integer variable indicating whether $D$ is selected, \\  ${z_D} \in \left\{ {0,1} \right\}$} \\
        $y_{uD}$ & \tabincell{l}{integer variable indicating whether $u$ is covered\\ by $D$, ${y_{uD}} \in \left\{ {0,1} \right\}$} \\
        $p(D)$ & power of $D$ \\
        $k(D)$ & capacity of $D$\\
        $\mathcal{D}(a)$ & set of disks centered on $a$ \\
        $d_D$  & number of TDs contained in $D$ \\
        ${\hat k_D}$  & \tabincell{l}{capacity of $D$ and varieties during the process} \\
        ${\hat p_D}$  & \tabincell{l}{power increment of $D$ and varieties during the\\ process} \\
        $l_a$  & latest broken disk for $a$ \\
        $C_a$  & set of TDs covered by $a$ \\
        $e_D$  & local ratio of $D$ \\
        $D^*$  & disk with minimum local ratio \\
        $a^*$  & central AP of $D^*$ \\
        $D_{rmv}$  & set of disks that need to be removed \\
        \hline

    \end{tabular}
\end{table}
\section{A Local Ratio Approach for MPCC}
A local ratio approach is any algorithm that follows the problem-solving heuristic of making the locally optimal choice in each stage\cite{Cormen1990}. It is a greedy-based strategy. In many problems, the greedy strategy does not produce an optimal solution, but a greedy heuristic can yield locally optimal solutions that approximate a globally optimal solution in a reasonable amount of time. Since the MPCC problem is NP-hard, it is a considerable challenge to obtain an optimal solution for integer programming (\ref{ip:2}). Therefore, we propose a local ratio approach named \emph{minimum local ratio} (MLR), which can obtain an approximate solution for MPCC.

In a practical scenario, when an AP cannot serve two or more TDs that have the same distance from it simultaneously, the AP will determine which one to serve first based on factors such as the TD’s request sequence or the urgency of the task. How to determine which TD to serve in this case is not the focus of this paper. Thus, we present definition \ref{d1} to define the power partial order relationship between two disks. After the partial order relationship between the two disks with the same radius determined by the two TDs is clarified, who the AP
should serve first becomes evident. Map the positions of all TDs and APs to a two-dimensional coordinate system, and use $\cos (\overrightarrow {au} )$ to represent the cosine of the angle between vector $\overrightarrow {au}$ and the $x$-axis formed by $a$ and $u$.

\begin{definition} \label{d1}
The distance between TDs $u$ and $v$ from AP $a$ is the same; that is, $r({D_{au}}) = r({D_{av}})$. If $\cos (\overrightarrow {au} ) > \cos (\overrightarrow {av} )$, then $r({D_{au}}) \succ r({D_{av}})$, and the same $p({D_{au}}) \succ p({D_{av}})$. Although $r({D_{au}}) = r({D_{av}})$, ${D_{au}}$ contains $v$ but ${D_{av}}$ does not contain $u$.
\end{definition}

\subsection{Minimum Local Ratio Algorithm}
In Algorithm \ref{alg:MLR}, we describe the overall logic of the MLR algorithm in the form of pseudocode. Its input is an MPCC problem instance $(U,A,k)$, a set of disks $ \mathcal{D}$ and a mapping of power $p$. $\mathcal{D}$ is a candidate set from which a solution is created. After the input, we define some key variables. $d_D$ records the number of TDs that $D$ can currently contain (as the algorithm progresses, $d_D$ changes). Let ${\hat k_D}$ denote the capacity of $D$, which is initialized to $k(D)$. This capacity may change with the continuous allocation of TDs. The variable ${\hat p_D}$ is the power increment of $D$, which is initialized to $p(D)$. We define $l$ as the set of disks finally selected for all APs, where $l_a$ is the disk finally selected for $a$, so $l = \bigcup\nolimits_{a \in A} {{l_a}} $. Similarly, $C_a$ records the set of TDs that are eventually covered by $a$, and $C = \bigcup\nolimits_{a \in A} {{C_a}}$. After Algorithm \ref{alg:ChooseDisks}, we can obtain the results of $l$ and $C$, and the final result ${ \mathcal{ \bar D}}$ is the set of non-null values in $l$.

In the main loop of Algorithm \ref{alg:ChooseDisks}, each iteration calculates the local ratio $e_D$ corresponding to $D \in {\mathcal{D}}$ in the current iteration state. On the basis of lemma \ref{lemma1}, the disk ${D^*}$ that has the minimum local ratio in each iteration satisfies the inequality ${d_{{D^*}}} \le {\hat k_{{D^*}}}$. When ${D^*}$ is selected, we update the value of ${l_{{a^*}}}$ to $D^*$, where $a^*$ is the center of $D^*$, and then assign the TDs contained by $D^*$ in the current iteration to $a^*$ for coverage.

\begin{figure}[!t]

  \renewcommand{\algorithmicrequire}{\textbf{Input:}}
  \renewcommand{\algorithmicensure}{\textbf{Output:}}
  \begin{algorithm}[H]
    \caption{MLR}
    \label{alg:MLR}
    \begin{algorithmic}[1]
      \Require An instance of MPCC ${\text{(}}U,A,k{\text{)}}$, a set of disks ${\mathcal{D}}$, a mapping of power $p:{\mathcal{D}} \mapsto {\mathbb{R}^ + }$.
      \Ensure A set of selected disks $\mathcal{\bar D}$ and the set of coverage results for all APs $C$.
      \State ${d_D} \leftarrow \lvert D \rvert, \forall D \in {\mathcal{D}}$.
      \State ${\hat k_D} \leftarrow k(D),\forall D \in {\mathcal{D}}$.
      \State ${\hat p_D} \leftarrow p(D),\forall D \in {\mathcal{D}}$.
      \State $l,C \leftarrow {\text{ChooseDisks}}(U,d,\hat k,\hat p)$.
      \For{each $a \in A$}
      \If{${l_a} \ne \emptyset $}
      \State ${ \mathcal{\bar D}} \leftarrow { \mathcal{\bar D}} \cup {l_a}$.
      \EndIf
      \EndFor
      \State\Return ${ \mathcal{\bar D}},C$.
    \end{algorithmic}
  \end{algorithm}
\end{figure}

\begin{figure}[!t]
  \renewcommand{\algorithmicrequire}{\textbf{Input:}}
  \renewcommand{\algorithmicensure}{\textbf{Output:}}
  \begin{algorithm}[H]
    \caption{ChooseDisks}
    \label{alg:ChooseDisks}
    \begin{algorithmic}[1]
      \Require $U,d,\hat k,\hat p$.
      \Ensure $l,C$
      \State ${l_a} \leftarrow \emptyset,\forall a \in A.$($l_a$ is the latest broken disk for $a$)
      \State ${C_a} \leftarrow \emptyset,\forall a \in A$.($C_a$ is the set of TDs covered by $a$)
      \While {$U \ne \emptyset $}
      \State ${e_D} \leftarrow {\hat p_D}/\min ({\hat k_D},{d_D}),\forall D \in {\mathcal{D}}$.
      \State ${D^*} \leftarrow \arg \min ({e_D}:D \in {\mathcal{D}})$(${d_D} \le {\hat k_D}$ is established).
      \State ${a^*} \leftarrow {\text{the}}\;{\text{central}}\;{\text{AP}}\;{\text{of}}\;{D^*}$.
      \State ${l_{{a^*}}} \leftarrow {D^*}$.
      \State ${C_{{a^*}}} \leftarrow {C_{{a^*}}} \cup {D^*}$. (Make TDs in ${D^*}$ covered by $a^*$)
      \If {${d_{{D^{\text{*}}}}} = {\hat k_{{D^*}}}$}
      \State ${{\mathcal{D}}_{rmv}} \leftarrow \{ D:D \in {\mathcal{D}}({a^*})\}$.
      \Else {${d_{D^*}} < {\hat k_{D^*}}$}
      \State ${{\mathcal{D}}_{rmv}} \leftarrow \{ D:D \in {\mathcal{D}}({a^*})\;{\text{and}}\;r(D) \prec r({D^*})\}  \cup \{ {D^*}\} $.
      \EndIf
      \State ${\mathcal{D}} \leftarrow {\mathcal{D}} \backslash {{\mathcal{D}}_{rmv}}$.
      \State ${\hat p_D} \leftarrow {\hat p_D} - {e_{{D^*}}}\min ({\hat k_D},{d_D}),\forall D \in {\mathcal{D}}$.
      \For {each ${u^*} \in {D^*}$}
      \State $U \leftarrow U\backslash \left\{ {{u^*}} \right\}$.
      \State $D \leftarrow D\backslash \left\{ {{u^*}} \right\},\forall D \in {\mathcal{D}}:{u^*} \in D$
      \State ${\hat k_D} \leftarrow {\hat k_D} - 1,\forall D \in {\mathcal{D}}({a^*})$
      \If {${d_D} = 0$ or ${\hat k_D} = 0,\forall D \in {\mathcal{D}}$}
      \State ${\mathcal{D}}: = {\mathcal{D}} \backslash \{ D\} $.
      \EndIf
      \EndFor
      \EndWhile
      \State \Return $l,C$.
    \end{algorithmic}
  \end{algorithm}
\end{figure}

\begin{lemma} \label{lemma1}
The disk ${D^*}$ selected in each iteration in Algorithm \ref{alg:ChooseDisks} must satisfy ${d_{{D^*}}} \le {\hat k_{{D^*}}}$.
\end{lemma}

\begin{proof}
Assuming that the ${D^*}$ selected in an iteration satisfies ${d_{{D^*}}} > {\hat k_{{D^*}}}$, we know that ${e_{{D^*}}}= {\hat p_{{D^*}}}/{\hat k_{{D^*}}}$. Then, there must be a disk ${D^ - },p({D^ - }) \prec p({D^*})$, that happens to have ${d_{{D^ - }}} = {\hat k_{{D^ - }}}$. In the current iteration ${\hat p_{{D^ - }}} \prec {\hat p_{{D^*}}}$ and ${e_{{D^ - }}} \prec {e_{{D^*}}}$, which is contrary to the hypothesis.
\end{proof}

In the second part of the main loop of Algorithm \ref{alg:ChooseDisks}, all the relevant variables are updated. If  ${d_{{D^{\text{*}}}}} = {\hat k_{{D^*}}}$ in the current iteration, all the disks in $\mathcal{D}$ with $a^*$ as the center are removed. If ${d_{{D^{\text{*}}}}} < {\hat k_{{D^*}}}$, then there are remaining IP resources for $a^*$; we remove these disks in $\mathcal{D}$, with $a^*$ as the center and a smaller radius, and ${D^*}$ itself. For line 12 of Algorithm \ref{alg:ChooseDisks}, the power increment of all remaining disks in $\mathcal{D}$ is updated. We remove all the TDs allocated for the current iteration from the related set in the \texttt{for} loop. We also remove all disks from $\mathcal{D}$ where the number of TDs contained in the disk is 0 or its capacity is 0. Finally, $l$ can record the disk for each AP (an AP may not choose any disk).

\begin{theorem}
\emph{MLR} computes a feasible solution to ILP \ref{ip:2} in polynomial time.
\label{theo_correctness}
\end{theorem}
\begin{proof}[proof (Correctness)]
\emph{MLR} outputs a feasible solution for ILP \ref{ip:2} because the \texttt{while} loop body in \emph{ChooseDisks} guarantees all TDs can be covered by APs. At first, according to Lemma \ref{lemma1}, the disk $D^*$ selected in each iteration can guarantee the capacity to cover the TDs contained in $D^*$. Second, because the capacities of the disks with the same center AP change together, the latest disk $D^{last}$ of $a\in A$ can cover the TDs covered previously. In other words, when $D^{last}$ is selected, it can also cover the TDs previously assigned to $a$. The former two points can satisfy constraints (\ref{ip:2.2}), (\ref{ip:2.3}) and (\ref{ip:2.4}). We assume that the total capacity of the system can cover all TDs, and $\mathcal{D}$ is formed based on all TDs and APs. In theory, any TDs can be covered by any AP. This guarantees that MLR can satisfy constraint (3).

\emph{(Polynomial Running Time):} The \texttt{while} loop iterates at most $n$ times to select a disk for each AP in Algorithm \emph{ChooseDisks}. Lines 4-8 take $O(mn)$ time to obtain a disk with the minimum local ratio. Lines 9-14 can be completed in $O(m)$ time. Then, line 15 updates the MLR of every $D \in \mathcal{D}$ in $O(mn)$ time. The body of the \texttt{for} statement (lines 16-23) takes $O(mn^2)$ time to update the variables of the system state and remove the incapacity disks from $\mathcal{D}$. In conclusion, MLR runs in polynomial time ($O(mn^3)$).
\end{proof}

\subsection{An Instance for MPCC}
\begin{figure}[!t]
\centering
\includegraphics[width=8cm]{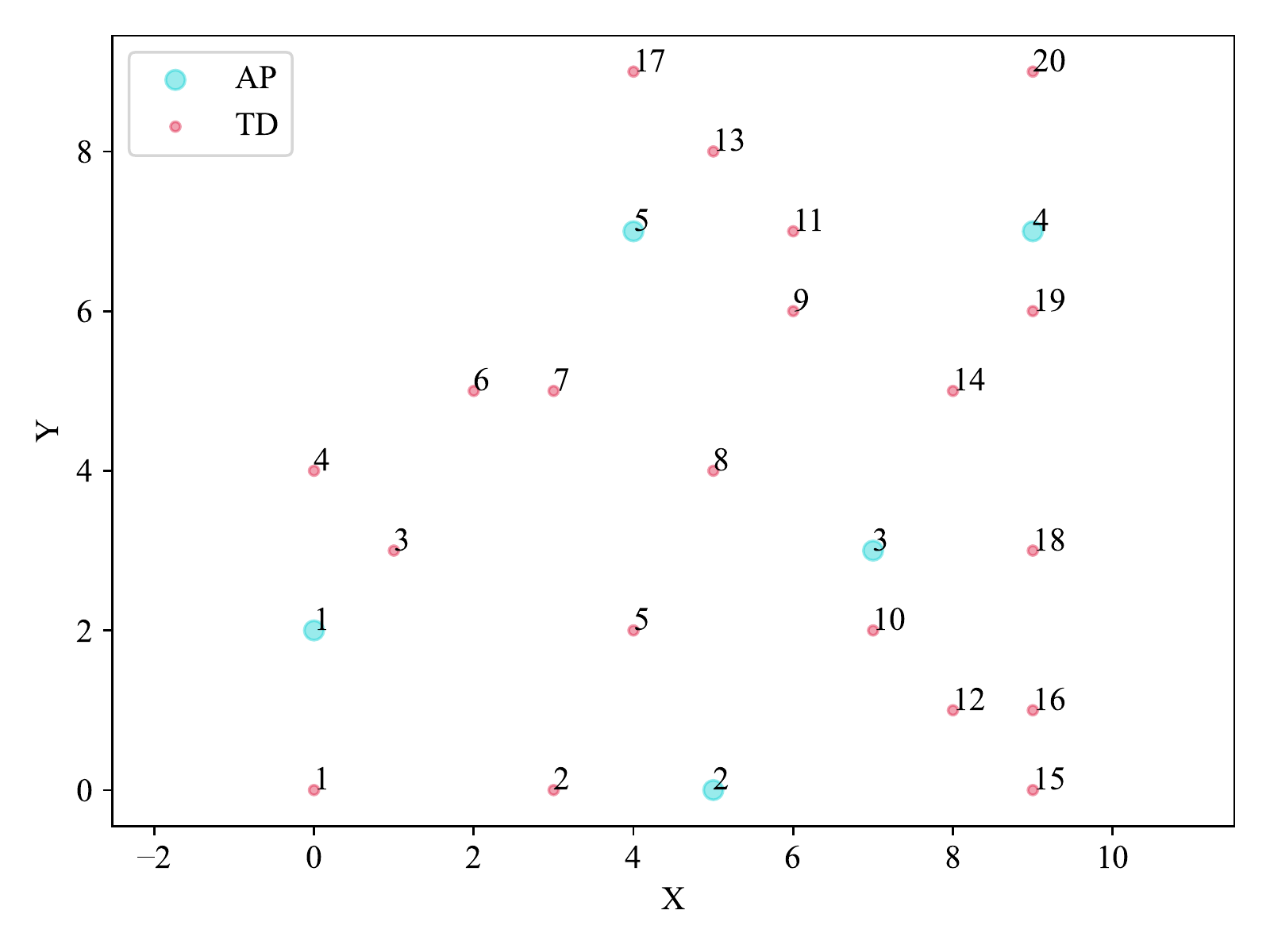}
\caption{An instance of using the MLR method to solve the MPCC problem}
\label{fig:ins}
\end{figure}
To further understand the MLR approach, we cite an instance to illustrate the process of it. In Fig. \ref{fig:ins}, we present an instance of the MPCC problem $(U, A, k)$, where $U = \{ 1,2,...,20\} $, $A = \{ 1,2,3,4,5\} $ and $k=5$. For ease of presentation, we abstract all facilities as dots distributed in 2-dimensional coordinates. The larger blue dots represent APs, and the smaller red dots represent TDs, as shown in Fig. \ref{fig:ins}. The disk drawn as a solid line is the final disk result obtained by the MLR method. The specific power value that each AP should provide can be calculated by the equation $p(a) = c \cdot r{(a)^\alpha }$ (in this case $c=1$, $\alpha=2$). The disk drawn as the dotted line is the disk temporarily selected by the MLR method during processing (that is, the value assigned to the variable $l_a$ during the algorithm's execution). There are nine disks  in Fig. \ref{fig:ins}, indicating that the main loop of the MLR method has been executed nine times in total. Next, we will introduce step by step how MLR obtained the results in this instance.

After inputting the instance $(U,A,k)$ and initializing the relevant variables, the algorithm enters the main loop for the first iteration. After calculating the local ratio $e$ of all disks in $\mathcal{D}$, the disk ${D_{3,10}}$ with the minimum local ratio is selected. In this iteration, $l_3$ is temporarily assigned as ${D_{3,10}}$, and TD10  is allocated to AP3. We remove ${D_{3,10}}$ from $\mathcal{D}$ and update local ratio of the remaining disks in $\mathcal{D}$ with ${e_{{D_{3,10}}}}$. Because TD10 is covered, at the end of the current iteration, the entire system must be updated. In the following iterations, this process is repeated continuously. According to the iteration sequence, the disks selected for each iteration are ${D_{3,10}}$, ${D_{4,19}}$, ${D_{5,7}}$, ${D_{3,8}}$, ${D_{1,4}}$, ${D_{2,5}}$, ${D_{4,20}}$, ${D_{2,16}}$ and ${D_{1,6}}$. After the algorithm's main loop ends, the result set of the algorithm is the union of all ${l_a},a \in A$, that is, $ \mathcal{\bar D} = \{ {D_{1,6}}, {D_{2,16}}, {D_{3,8}}, {D_{4,20}}, {D_{5,7}}\}$. In addition, the coverage information of all APs is also obtained: ${C_1} = \{ 1,3,4,6\} $, ${C_2} = \{ 2,5,15,16\} $, ${C_3} = \{ 8,10,12,14,18\} $, ${C_4} = \{ 19,20\}$ and ${C_5} = \{ 7,9,11,13,17\}$. These TDs, which are simultaneously contained by some disks in $\mathcal{\bar D}$, will be covered by the disk selected first. For example, ${D_{3,8}}$ and ${D_{2,16}}$ contain TD8,10 and 12 in Fig. \ref{fig:ins}. During the algorithm’s execution, ${D_{3,8}}$ is selected before ${D_{2,16}}$, so ${D_{3,8}}$ covers that TD8, 10, 12.

\section{Experimental Results}
According to the MLR algorithm proposed above, we use synthetic data to conduct practical experiments to simulate the power control for APs in the wireless network. The experiments ignore the vertical distribution of two facilities but map their positions to a 2-dimensional coordinate system. The relevant parameters are shown in Table \ref{table:param}. The specific experimental settings are as follows:

\begin{enumerate}
  \item The hardware configuration of the experimental environment is as follows: the CPU is an Intel i7-10700, configured with 8 cores and 16 threads at 2.9 GHz, with 16 GB memory, and a hard disk capacity of 1 TB.
  \item The entire system includes two facilities, TD and AP, which are distributed randomly in a 2-dimensional space.
  The capabilities of each AP are limited, and all TDs must be covered by an AP. To ensure that the total system capacity meets the requirements of covering all TDs, we guarantee $m \cdot k \ge n$.
  \item Each experiment is performed 50 times, and the final results are averaged to reduce the impact of randomness.
  \item  The IP utilization rate of an AP is the ratio of the number of TDs covered by it to its capacity.
  \item The variance of IP utilization is defined by formula (\ref{eq:variance})
      \begin{equation}
        {s^2} = \frac{{\sum\nolimits_{a \in A} {{{\left( {\lvert {{C_a}} \rvert - n/m} \right)}^2}} }}{m}.
        \label{eq:variance}
      \end{equation}
      According to the property of variance, the smaller the value of ${s^2}$ is, the more balanced the number of TDs covered by each AP; on the contrary, a few APs cover all TDs.
  \item In this section, we compare the MLR with the OPT and NCA (nearest capable access) approaches. OPT uses IBM's open source tool CPLEX to obtain the optimal solution of the MPCC problem. NCA is another greedy-based method used to solve the MPCC problem. When deciding which AP will cover which TD in each iteration, it selects the closest AP-TD point pair where the AP still has
       capacity.
\end{enumerate}

\begin{table}[]
\caption{Configuration of experimental parameters}
\centering
\begin{tabular}{lll}
\hline
\textbf{Param} & \textbf{Description} & \textbf{Value} \\ \hline
$\alpha$ & Power parameter in equation \ref{eq:1} & [1,5] \\
$c$ & Power parameter in equation \ref{eq:1} & 1 \\
$k$ & Capacity of AP & [25,100] \\
$r$ & Side length of the area & [15,100] \\
$p^X_{iAP}$ ($p^X_{jTD}$) & X coordinate of AP i (TD j) & [0,100] \\
$p^Y_{iAP}$ ($p^Y_{jTD}$)& Y coordinate of AP i (TD j) & [0,100] \\ \hline
\label{table:param}
\end{tabular}
\end{table}

\subsection{Impact of the Number of TDs}
\begin{figure*}[!t]
\centering
\subfigure[Total power] {
\includegraphics[width=1.8in]{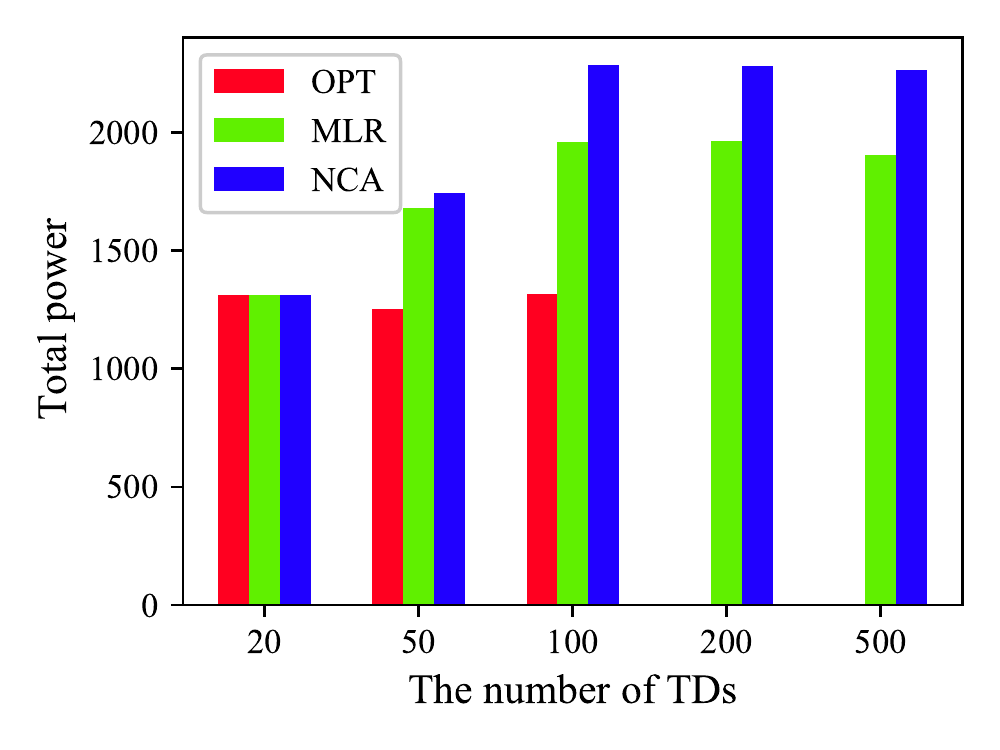}
\label{fig:case1:a}
}
\hfil
\subfigure[Execution time] {
\includegraphics[width=1.8in]{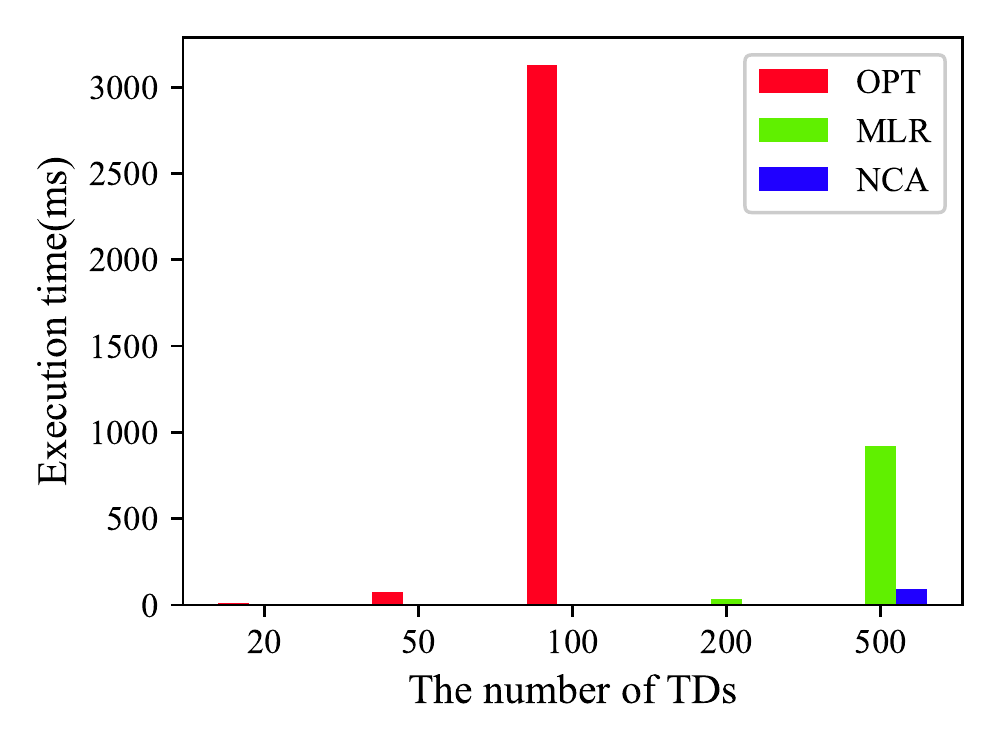}
\label{fig:case1:b}}
\hfil
\subfigure[Variance of IP utilization] {
\includegraphics[width=1.8in]{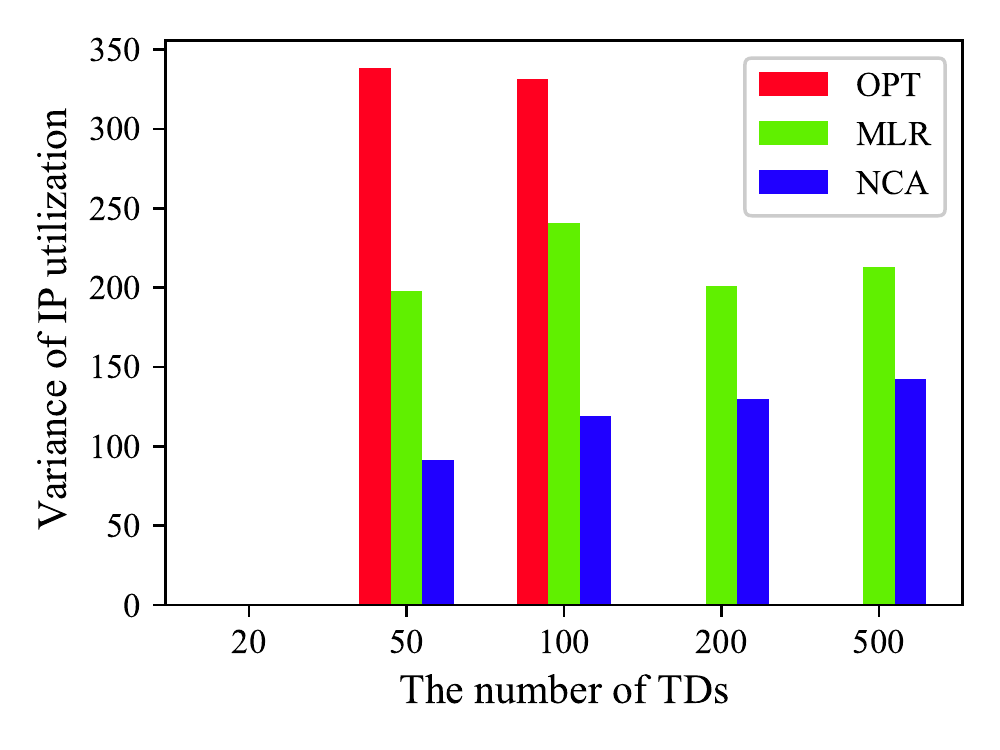}
\label{fig:case1:c}}
\caption{System performance under different numbers of TDs.}
\label{fig:case1}
\end{figure*}
In this experiment, we analyzed the impact of changes in the number of TDs on the MPCC problem. The main foci are  the total system power, algorithm execution time, and AP capacity utilization variance. The number of TDs gradually increases from 20 to 500, and all facilities are distributed in an area with a side length of 40. The ratio between the number of TDs and APs is maintained at 25:1, and the capacity of each AP is $ k = 40$. For the two constants in equation (\ref{eq:1}), the values are $c=1$ and $\alpha {\text{ = 2}}$. When the number of TDs is greater than 200, the optimal solution of a single instance cannot be obtained within 10 minutes.

Fig. \ref{fig:case1:a} shows the variation in the total power of the three algorithms as the number of TDs increases. Except that the power of the three methods is the same when the number of TDs is 20 (there is only one AP at this time), MLR is significantly better than NCA in all other cases. When the number of TDs changes in the range of 20 to 100, the total power of MLR and NCA increases significantly, while the change in OPT is modest because all facilities are randomly and evenly distributed in a fixed-size area. When the numbers of TDs and APs are small, the result of the optimal solution will not produce a large difference in the total power under different conditions. For MLR and NCA, in some cases, the pursuit of local optimization will produce worse results globally. When the number of TDs increases from 100 to 500, the  area remains a fixed size. The total power of MLR and NCA changes only slightly, and the total power of MLR actually decreases. When the number of TDs is small (100), their distribution is sparse, and APs require a larger radius to ensure coverage. When the number of TDs is large (500), the distribution of them and the APs is dense. Although more APs are activated, their coverage radii are reduced. Therefore, under the condition of considering only power (not considering the cost of APs) and a dense distribution of TDs, even if fewer APs can cover all the TDs, more APs have reduced power.

Fig. \ref{fig:case1:b} shows that the time required to obtain the optimal solution increases exponentially as the number of terminal devices increases. When the scale is small, both MLR and NCA can produce results in a very short time. Since MLR can still produce better results than NCA in 10 s at the maximum scale, this approach is acceptable in actual application scenarios. As shown in Fig. \ref{fig:case1:c}, in the results of OPT, AP coverage of TDs is more concentrated, while NCA is the most dispersed, and the result of MLR is somewhere in between.
\subsection{Impact of ${k}$}
In this experiment, we explored the impact of variations in the capacity of APs under the condition that the numbers of APs and TDs remain stable. Where $m = 4$ and $n = 100$, all facilities are randomly and evenly distributed in an area with a side length of 40. The capacity of AP $ k$ varies from 25 to 100. Clearly, when $ k = 25$, all APs can just meet the requirements of covering all TDs.

\begin{figure*}[!t]
\centering
\subfigure[Total power] {
\includegraphics[width=1.8in]{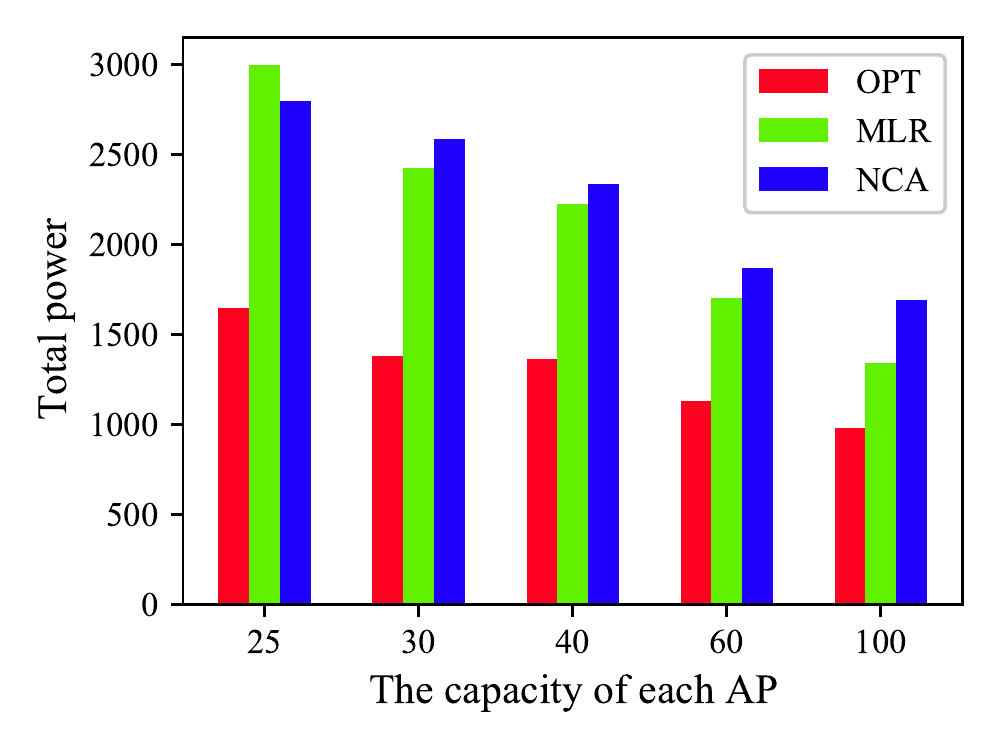}
\label{fig:case2:a}}
\hfil
\subfigure[Execution time] {
\includegraphics[width=1.8in]{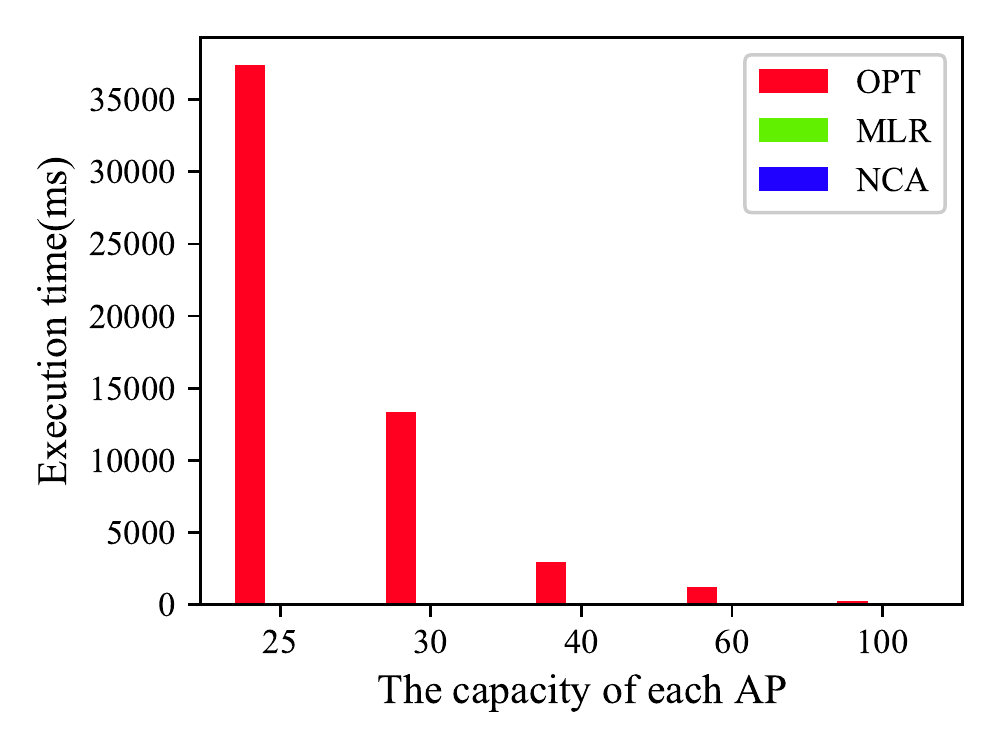}
\label{fig:case2:b}}
\hfil
\subfigure[Variance of IP utilization] {
\includegraphics[width=1.8in]{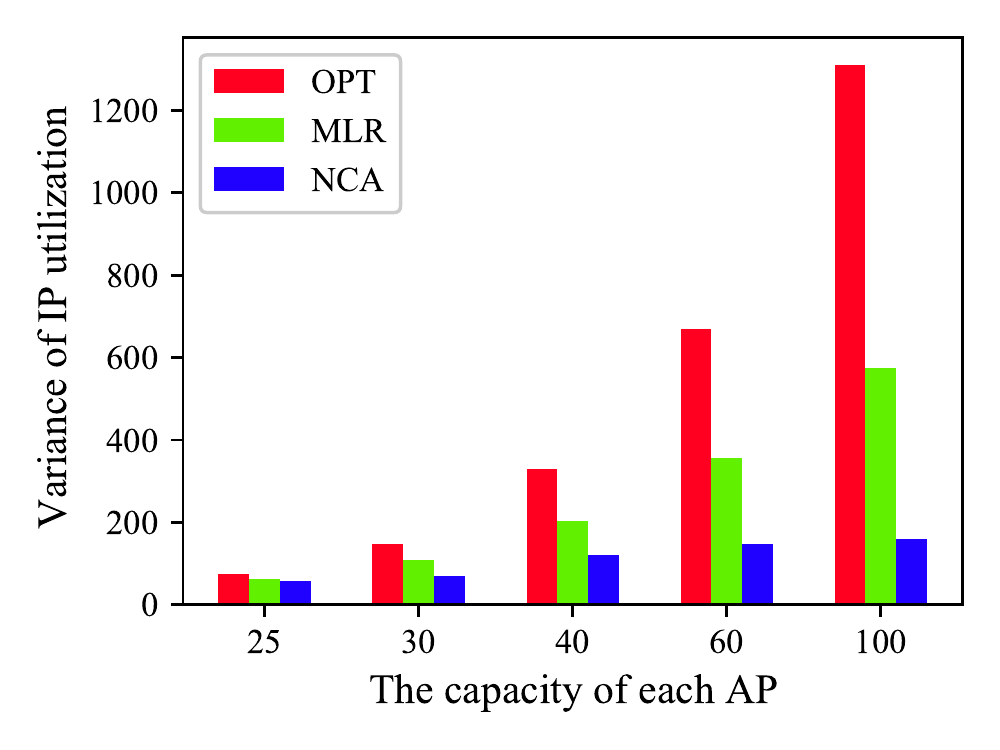}
\label{fig:case2:c}}
\caption{System performance under different ${k}$.}
\label{fig:case2}
\end{figure*}

Fig. \ref{fig:case2} shows the total system power, execution time and the variance of IP utilization as $ k$ changes. MLR's performance is not as good as NCA when $ k = 25$, as shown in Fig. \ref{fig:case2:a}. In highly tight capacity conditions, the MLR method will have more extreme situations, such as the last few TDs being covered by APs that are far away. However, as $ k$ continues to increase, the total power obtained by MLR decreases faster. In practical applications, the extreme situation of capacity shortage is a small probability event (if it occurs frequently, it means that additional APs need to be added). MLR shows better performance when there is surplus capacity. From Fig. \ref{fig:case2:b}, we know that the OPT method consumes a lot of time when capacity is tight. In contrast, the execution time of the MLR method and NCA is not related to $ k$. The variance results shown in Fig. \ref{fig:case2:c} indicate that the distribution of TDs tends to become more dense as $ k$ increases. Thus, when r = 40, selecting a disk with a larger radius that covers more TDS can save more power. Section 4.3 discusses the influence of the number of APs on the variance under the condition of constant $ k$.

\subsection{Impact of the number of APs with ${k} = 25$}
In the third experiment, we kept the number of TDs $n$ and the capacity of APs $ k$ unchanged, with $n=100$ and $ k=25$. In this case, as $m$ increases, the total capacity of the entire system also increases. $m$ gradually grows from 4 to 20, and the corresponding total system capacity increases from 100 to 500.

\begin{figure*}[!t]
\centering
\subfigure[Total power] {
\includegraphics[width=1.8in]{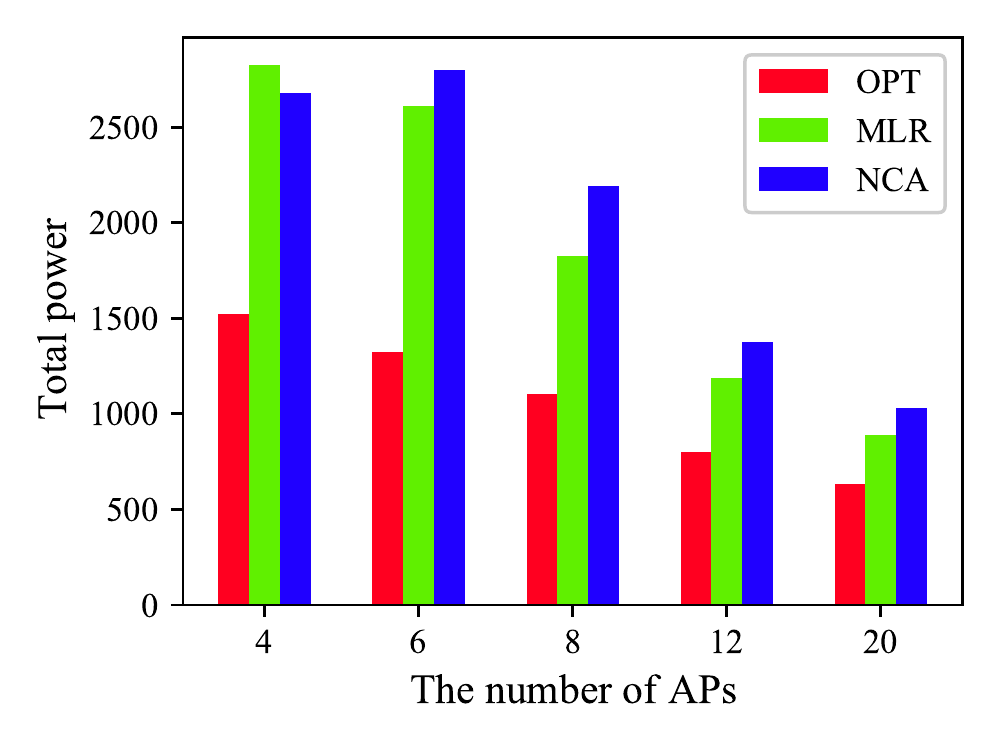}
\label{fig:case3:a}}
\hfil
\subfigure[Approximation ratio] {
\includegraphics[width=1.8in]{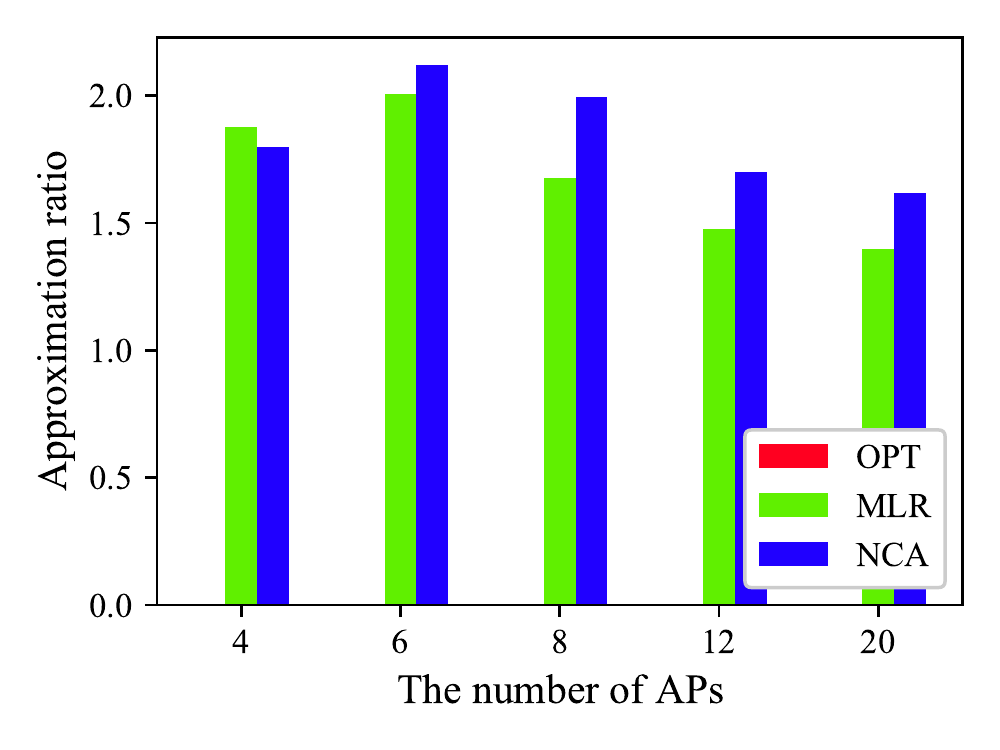}
\label{fig:case3:b}}
\hfil
\subfigure[Variance of IP utilization] {
\includegraphics[width=1.8in]{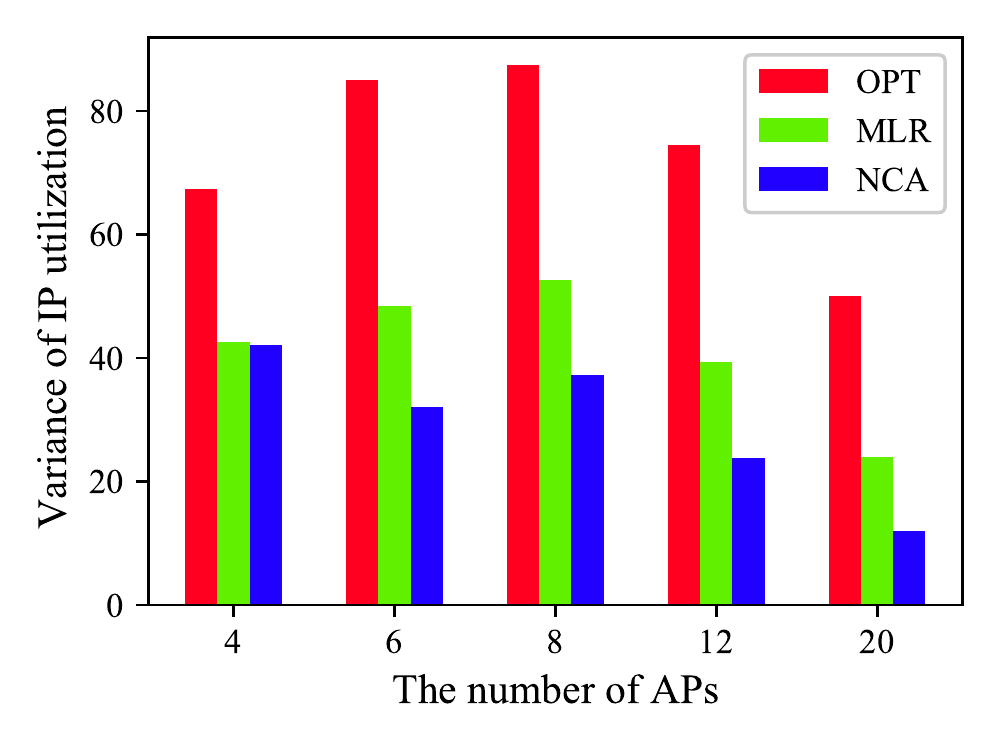}
\label{fig:case3:c}}
\caption{System performance under different numbers of APs (${k}=25$).}
\label{fig:case3}
\end{figure*}

\begin{figure*}[!t]
\centering
\subfigure[Total power] {
\includegraphics[width=1.8in]{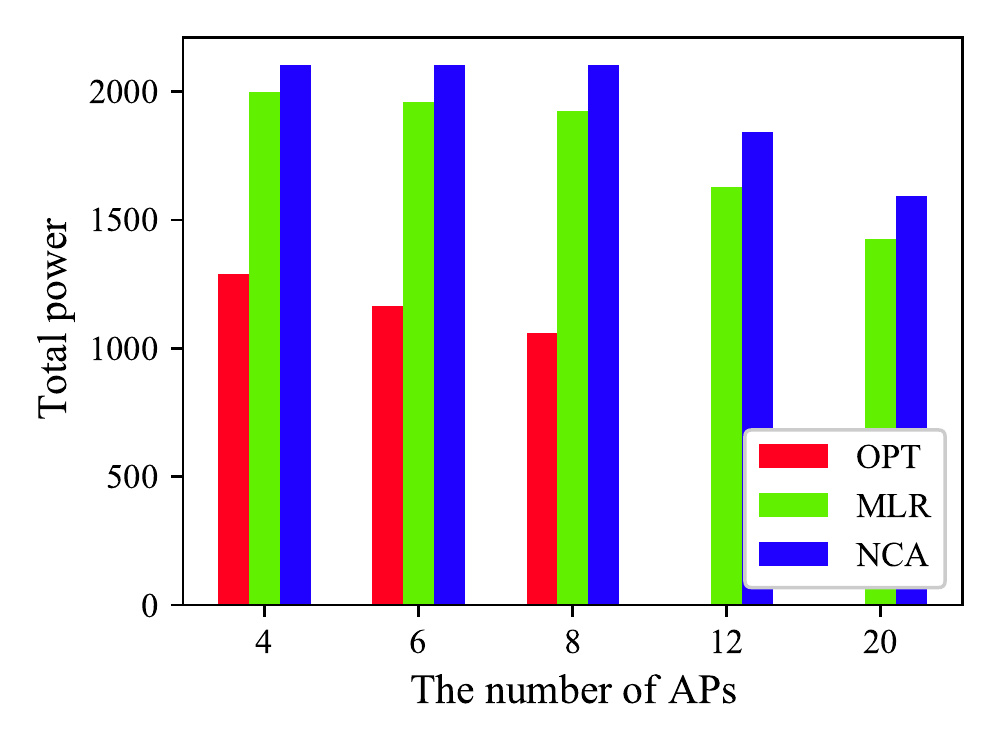}
\label{fig:case4:a}}
\hfil
\subfigure[Execution time] {
\includegraphics[width=1.8in]{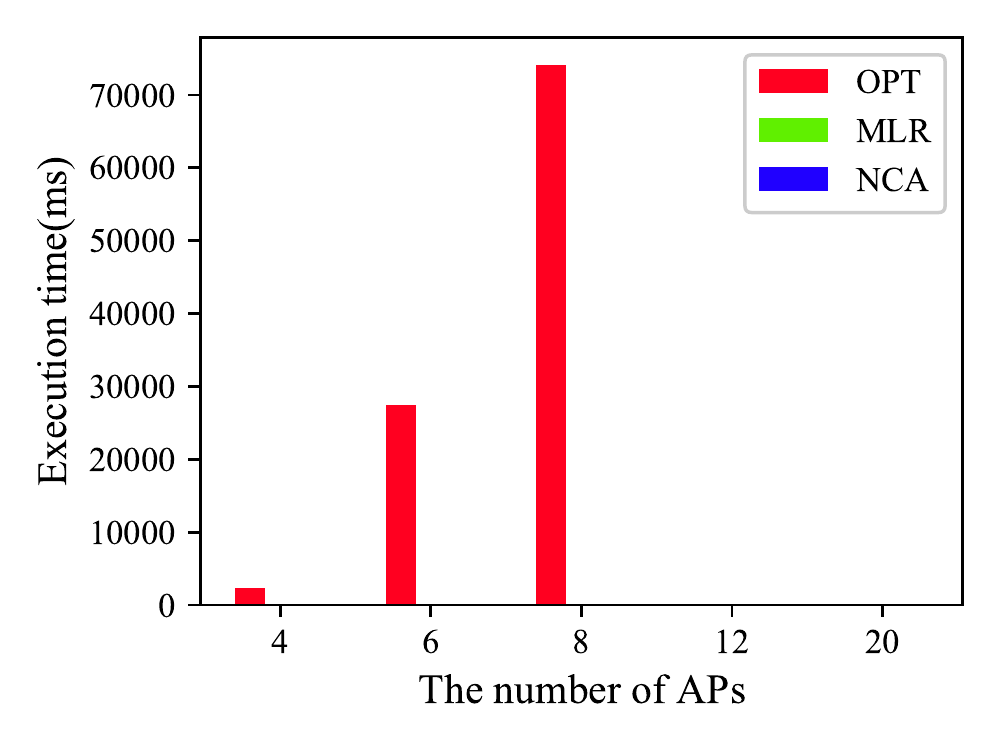}
\label{fig:case4:b}}
\hfil
\subfigure[Variance of IP utilization] {
\includegraphics[width=1.8in]{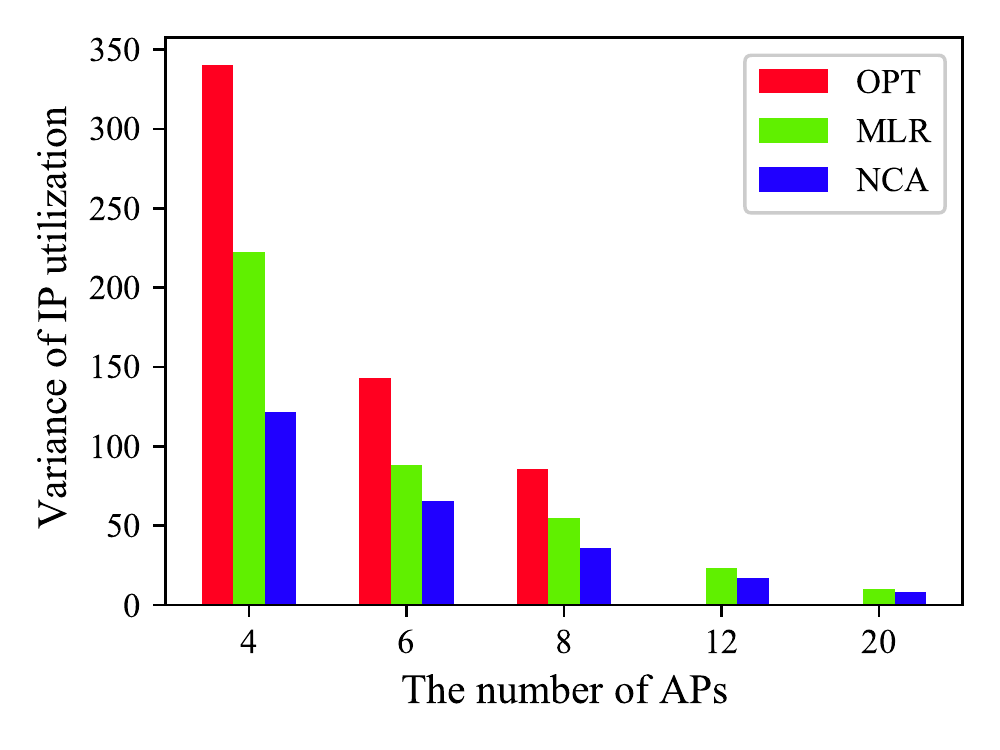}
\label{fig:case4:c}}
\caption{System performance under different numbers of APs ($K=160$).}
\label{fig:case4}
\end{figure*}

When $m=4$, the system capacity is just sufficient to cover all TDs. In this case, as shown in Fig. \ref{fig:case3:a}, the results obtained by MLR and NCP are worse than the optimal solution. Notably, the MLR method is not better than NCP because when capacity is incredibly tight, if some TDs are densely distributed around an AP, the MLR method will fill it more quickly and fall into a local optimum. For example, in Fig. \ref{fig:ins}, after AP3 covers TDs $\{ 8,10,12,14,18\} $, AP2 has to provide more power to cover the TDs $\{ 15,16\}$ far away from it. As the number of APs and the total capacity of the system increase, the total power obtained by the three methods decreases. Fig. \ref{fig:case3:b} shows that the performance improvement of the MLR method is the most obvious. Fig. \ref{fig:case3:c} shows the variance changes of each method. The most obvious is that the coverage schemes of the OPT method are very intensive in all cases. It reaches the maximum value when m=8 (when m=4, n=100, the variance of OPT is still different from that of the other two methods because the total capacity K of the system may be greater than 100).

\subsection{Impact of the number of APs with $K=160$}
In this experiment, we kept the number of TDs $n$ and the total system capacity $K$ constant at 100 and 160, respectively. The number of APs $m$ gradually increases from 4 to 20, and $ k$ changes accordingly. All facilities are randomly and evenly distributed in an area with a side length of 40. When $m>12$, OPT cannot produce results within 10 minutes.

As shown in Fig. \ref{fig:case4:a}, the results of MLR are better than those of NCA. As the number of APs increases, the distribution of APs in a fixed-size area becomes denser. Overall, the average radius of APs will be reduced, and the impact of the decreasing radius on the system is greater than that of the number of APs (more disks). As the number of APs increases and the capacity decreases, the variance of IP utilization among APs becomes more similar, as shown in Fig. \ref{fig:case4:c}. In Fig. \ref{fig:case4:b},  the variation in the number of APs also has a considerable impact on the execution time of the OPT method. On the contrary, the other two have almost no fluctuations.

\section{Conclusion}
Signal coverage consumes considerable energy in wireless networks, so this paper studies how to assign an appropriate power for APs to reduce energy consumption. We consider the signal coverage process in edge networks as a minimum power coverage problem and define this problem as an MPCC problem, which is a special case of MPC problem. Capacity constraints make it more difficult to solve. Modeling this problem in MEC is also a novel exploration. Then, we propose a local-ratio-based approach to solve the MPCC problem. Finally, various experiments show that our method can produce satisfactory results in most cases.

For the MPCC problem, we can also learn from the previous methods of the MPC problem to solve it (such as the primal-dual method). The MPCC problem is more complex but more valuable when considering heterogeneous APs and TDs' natural number requirements (not just 1). This is worthy of further study. In addition, this paper notes in Experiment 4.3 that more APs can decrease the total energy of the system. This does not take into account the cost of the AP itself, which can be considered in future studies.
\section{Acknowledgments}
This work was supported in part by the National Natural Science Foundation of China [Nos. 12071417, 61762091, 62062065], in part by the 13th Postgraduate Innovation Project of Yunnan University [No. 2021Z079] and Scientific research foundation of Yunnan Provincial Department of Education [No. 2022J0002].

\bibliographystyle{ACM-Reference-Format}
\bibliography{ref}{}
\appendix

\end{document}